\documentclass[conference]{IEEEtran}

\ifCLASSINFOpdf
\else
   \usepackage[dvips]{graphicx}
\fi
\usepackage{url}

\usepackage{amsmath}
\usepackage{amssymb}
\usepackage{amsthm}
\usepackage{mathtools}
\usepackage{bbm}
\usepackage{bm}
\usepackage{multirow}
\usepackage{comment}
\usepackage{booktabs}
\usepackage[caption=false,font=footnotesize]{subfig}
\usepackage{glossaries}
\usepackage[hyperfootnotes=false]{hyperref}
\usepackage[dvipsnames]{xcolor}
\usepackage{enumitem}
\usepackage{float}
\usepackage[symbol]{footmisc}

\newcounter{mytempeqncnt}

\newacronym{ue}{UE}{user equipment}
\newacronym[plural=SPs,firstplural=scattering points (SPs)]{sp}{SP}{scattering point}
\newacronym{ofdm}{OFDM}{orthogonal frequency division multiplexing}
\newacronym{ula}{ULA}{uniform linear array}
\newacronym{awgn}{AWGN}{additive white Gaussian noise}
\newacronym{rcs}{RCS}{radar cross section}
\newacronym{aoa}{AOA}{angle of arrival}
\newacronym{aod}{AOD}{angle of departure}
\newacronym{toa}{TOA}{time of arrival}
\newacronym{crlb}{CRLB}{Cramér-Rao lower bound}
\newacronym{los}{LOS}{line of sight}
\newacronym{fdr}{FDR}{false discovery rate}
\newacronym{mac}{MAC}{medium access control}
\newacronym{isac}{ISAC}{integrated sensing and communication}
\newacronym{deb}{DEB}{delay error bound}
\newacronym{aeb}{AEB}{angle error bound}
\newacronym{rmse}{RMSE}{root mean square error}
\newacronym{omp}{OMP}{orthogonal matching pursuit}
\newacronym{music}{MUSIC}{multiple signal classification}
\newacronym{esprit}{ESPRIT}{estimation of signal parameters via rotational invariant techniques}
\newacronym{em}{EM}{expectation maximization}
\newacronym{ml}{ML}{maximum likelihood}
\newacronym{pdf}{PDF}{probability density function}
\newacronym{cdf}{CDF}{cumulative distribution function}
\newacronym{fwer}{FWER}{familywise error rate}

\def\C{\mathbb{C}}
\DeclareMathOperator{\arctantwo}{arctan2} 

\newcommand{\bs}[1]{\boldsymbol{#1}}
\DeclareMathOperator*{\argmax}{\arg\max}

\newtheorem{lemma}{Lemma}
\newtheorem{proposition}{Proposition}

\newcommand{\rev}[1]{\textcolor{black}{#1}}

\newcommand\blfootnote[1]{%
  \begingroup
  \renewcommand\thefootnote{}\footnote{#1}%
  \addtocounter{footnote}{-1}%
  \endgroup
}

\allowdisplaybreaks

\begin{document}

\title{Interference Detection and Exploitation\\for Multi-User Radar Sensing}
%
%
\author{%
\IEEEauthorblockN{Laurits~Randers\IEEEauthorrefmark{1}\IEEEauthorrefmark{4}, Martin~V.~Vejling\IEEEauthorrefmark{2}\IEEEauthorrefmark{3}\IEEEauthorrefmark{4}, Petar~Popovski\IEEEauthorrefmark{2}}
\IEEEauthorblockA{\IEEEauthorrefmark{1}Independent Researcher, Aalborg, Denmark (laurits.randers@gmail.com)}
\IEEEauthorblockA{\IEEEauthorrefmark{2}Dept. of Electronic Systems, Aalborg University, Denmark (\{mvv,petarp\}@es.aau.dk)}
\IEEEauthorblockA{\IEEEauthorrefmark{3}Dept. of Mathematical Sciences, Aalborg University, Denmark (mvv@math.aau.dk)}
\IEEEauthorblockA{\IEEEauthorrefmark{4}These authors contributed equally.}
}

\maketitle

\begin{abstract}
Integrated sensing and communication is a key feature in next-generation wireless networks, enabling joint data transmission and environmental radar sensing on shared spectrum.
In multi-user scenarios, simultaneous transmissions cause mutual interference on overlapping frequencies, leading to spurious target detections and degraded sensing accuracy. This paper proposes an interference detection and exploitation algorithm for sensing using spectrally interleaved orthogonal frequency division multiplexing. A statistically rigorous procedure is introduced to detect interference while controlling the familywise error rate. We propose an algorithm that estimates the angle by exploiting interference, while estimating the delay by avoiding the interference. Numerical experiments demonstrate that the proposed method reliably detects interference, and that the delay and angle estimation error approaches the Cram\'{e}r–Rao lower bound.
\end{abstract}



\section{Introduction}
\IEEEPARstart{I}{ntegration} of sensing and communication is a key technological advancement for the next generation of communication networks, and extensive attention have been paid to topics such as dual-functional sensing and communication waveform design, fundamental bounds on detection and sensing with existing or envisioned infrastructure, and practical algorithms to achieve these bounds \cite{Liu2022:Integrated,Liu2022:Fundamental,Gonzalez2024:Integrated}. A less studied aspect in the \gls{isac} literature is the problem of \gls{mac}, that is, multiple users sensing over a shared wireless medium. \textcolor{white}{\blfootnote{This work was partly funded by the Villum Investigator Grant ``WATER'' financed by the Villum Foundation, Denmark.}} 

Generally, wireless \gls{mac} protocols are either based on \emph{scheduling} where users are granted resources thereby guaranteeing no overlap, or \emph{random access} where users may cause collisions and interfere when occupying the same resources \cite{Gummalla2000:Wireless}. When user density is high, scheduling is always necessary, however, the communication overhead and delay incurred by scheduling can be limiting.
Collisions due to multi-user interference are fundamentally different in communication and sensing, respectively. In communication, the receiver can rely upon codes to detect that no correct message has been received, followed by transmission of a negative feedback to the sender, and they retry after a random time, as in ALOHA. In sensing, such code-based detection is not possible and collisions must be detected through variations in the received power \cite{Liu2024:Next}. We use this mechanism to detect and utilize interference, as shown through the next minimal example.
\begin{figure}[t]
    \centering
    \includegraphics[page=1, width=0.87\linewidth]{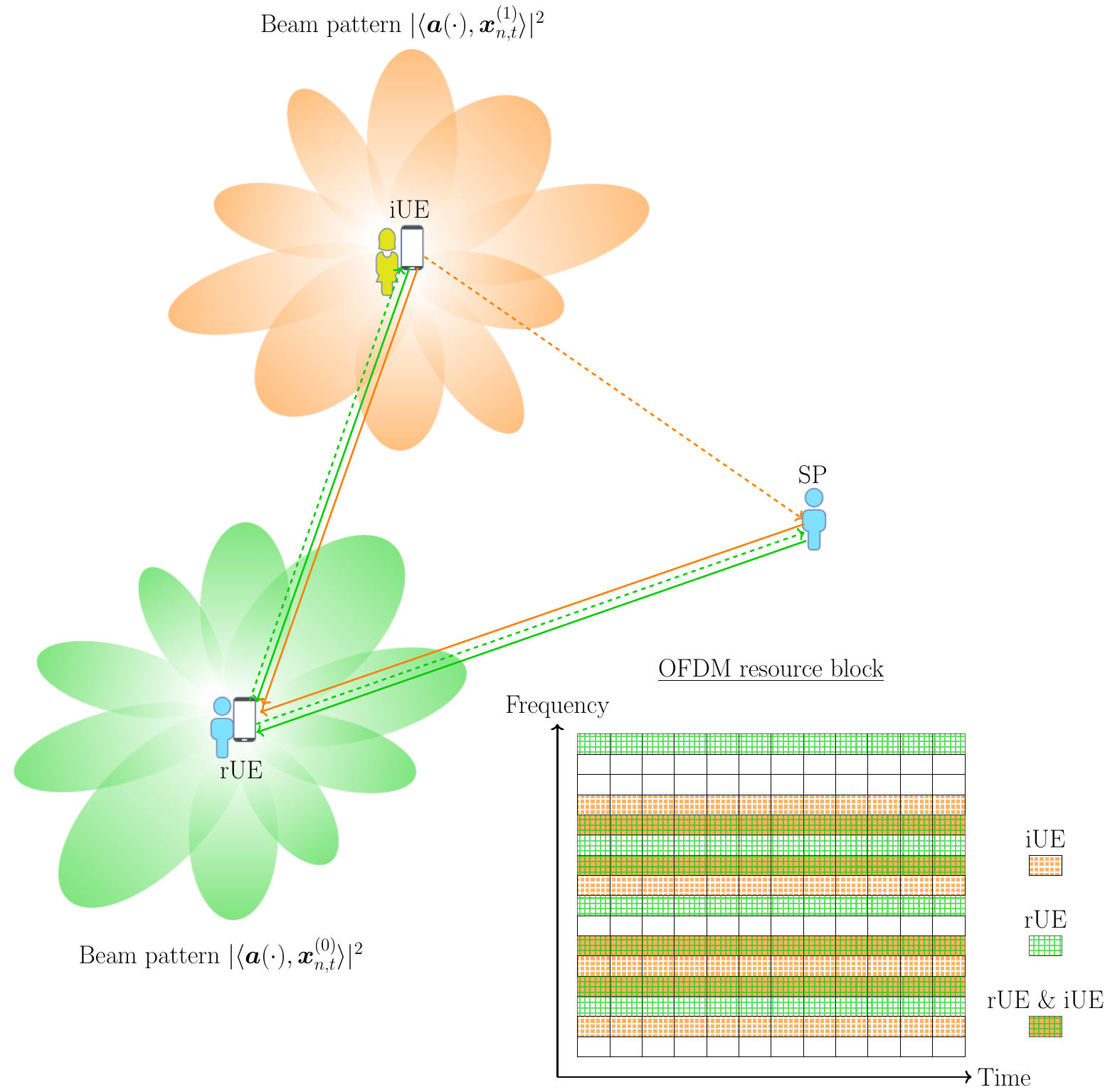}
    \caption{Illustration of the scenario. The orange and green lines depict signal paths of the iUE and rUE, respectively, with solid lines for signals impinging on the sensor. An OFDM resource block display how the iUE and rUE occupy distinct parts of the spectrum with overlap at some subcarriers.}
    \label{fig:scenario}
\end{figure}
%


An illustration of the considered system is given in Figure~\ref{fig:scenario}. A reference \gls{ue}, abbreviated rUE, performs monostatic sensing of the environment, consisting of an interfering \gls{ue}, abbreviated iUE, and a \gls{sp}. Both devices use \gls{ofdm} for sensing (and communication) on shared spectrum resources, and spectrally interleaved subcarriers are randomly allocated so that each device occupies only a subset of the total bandwidth. Without collaborating to allocate orthogonal spectrum resources, overlapping subcarriers may occur, causing missed detections or detection of spurious non-existing targets. Due to the strong \gls{los} path from the iUE to the rUE, interfered subcarriers exhibit higher received power, enabling their detection and exclusion for delay estimation. However, since the interference still carries angular information, we exploit the interfered resources in a secondary step for angle estimation, followed by data association.

Handling mutual interference in multi-user radar sensing systems can rely on two basic approaches~\cite{Niu2025:Interference}. In the first approach, the interfering signal is initially estimated and subsequently this is exploited for channel estimation \cite{Sit2018:Mutual}. While this can be effective, it, however, relies on estimating accurately the interfering signal, which is only possible under specific assumptions on the interfering signal. An information-theoretical analysis of collaborative two transmitter \gls{isac} systems is found in \cite{Liu2025:Fundamental}. The second approach detects which resources have interference and then removes those resources before channel parameter estimation \cite{Nuss2017:Novel}. This method is feasible without making any assumptions on the interfering signal, however, it is suboptimal and only effective if a relatively small proportion of the used resources have interference. In \cite{Nuss2017:Novel}, subcarriers in a spectrally interleaved \gls{ofdm} system with interference are detected by ad hoc thresholding of the received power.

We propose a novel interference detection and exploitation algorithm that reliably detects \gls{ofdm} resources with interference, and subsequently estimates delays using only the \gls{ofdm} resources without interference but estimates angles using all the \gls{ofdm} resources. For interference detection, we use the technique from \cite{Nuss2017:Novel}, and contribute with a theoretical justification for this approach as well as a hyperparameter selection method allowing for statistically rigorous \gls{fwer} control.

The main contributions can be summarized as follows:
\begin{itemize}
    \item We prove that the \gls{fwer} of the interference detection technique from \cite{Nuss2017:Novel} does not depend on the unknown channel.
    \item We propose an algorithm that estimates the angles using \gls{music} with all the \gls{ofdm} resources, estimates jointly the angles and delays using \gls{omp} with only the \gls{ofdm} resources without interference, and finally associates these two sets of estimated parameters.
    \item With a tailored Fisher analysis, we provide Cram\'{e}r-Rao lower bounds on the delay and angle estimators, thereby revealing that the information contained in the interfered resources is negligible except when estimating the angle to the interferer.
    \item Numerical experiments show that the squared error of the proposed delay and angle estimation algorithm nearly reaches the theoretical lower bound.
\end{itemize}

The remainder of the paper is organized as follows. The system model is presented in Section~\ref{sec:system_model}, followed by the proposed procedure in Section~\ref{sec:methods}.
In Section~\ref{sec:theoretical_performance_bounds}, we discuss the Fisher analysis tailored to the considered scenario. Numerical experiments are given and discussed in Section~\ref{sec:numerical_experiments}, and finally conclusions are stated in Section~\ref{sec:conclusion}.

\textit{Notation:} For $N\in \mathbb{N}$, we let $[N] = \{1,\dots,N\}$.
We denote by $\langle \cdot, \cdot\rangle$ and ${\rm Tr}(\cdot)$ the $L_2$ inner product and trace operator, respectively.
$(\cdot)^\top$, $(\cdot)^*$, and $(\cdot)^H$ denote transpose, complex conjugate, and hermitian, respectively.
We use $|\cdot|$ to denote either the absolute value of a complex number or the cardinality of a set depending on the context, and $\times$ to denote Cartesian product.
$\mathbbm{1}[\cdot]$ denotes the indicator function equal to one when the condition $\cdot$ is true, and zero otherwise.
For a matrix $\bs{A}$, $\bs{A}_{i,j}$ and $[\bs{A}]_{i,j}$ denote the entry in the $i$-th row and $j$-th column.
$\mathbb{C}^{N}$ denotes the space of $N$-dimensional complex vectors.


\section{System model}\label{sec:system_model}
The r\gls{ue} is located at the origin $\bs{p}_0=\bs{0}$, the i\glspl{ue} at $\bs{p}_i=[p_{i,x},p_{i,y}]^\top$, $i\in[O]$, and the passive \glspl{sp} at $\bs{p}_{O+l}=[p_{O+l,x},p_{O+l,y}]^\top$, $l\in[L]$. 
The \glspl{ue} operate in full-duplex \gls{ofdm} at carrier frequency $f_c$, such that the wavelength is $\lambda=c/f_c$ where $c$ is the speed of light. The subcarrier frequencies are $f_n=f_c+n\Delta_f$ for $n=0,\dots,N-1$, where $N$ is the number of subcarriers and $\Delta_f$ the spacing, yielding bandwidth $W=N\Delta_f$. The \glspl{ue} are equipped with \glspl{ula} of $N_{\rm u}$ antennas at $\lambda/2$ spacing. Note that cyclic prefix \gls{ofdm} is shown to be optimal for ranging \cite{Liu2025:OFDM}, and is also the implemented communication standard, making it attractive for wide deployment of \gls{isac}.

We denote by $\bs{x}_{n,t}^{(i)} \in \mathbb{C}^{N_{\rm u}}$ the transmitted signal on the $n$-th subcarrier and $t$-th \gls{ofdm} symbol for the $i$-th \gls{ue}.
Importantly, each \gls{ue} use spectrally interleaved subcarriers \cite{Sturm2013:Spectrally}, that is, they each use only a fraction of the total resources which is captured by the support set of $\bs{x}_{n,t}^{(i)}$. Specifically, we denote the active time-frequency resources of the $i$-th \gls{ue} by
$\Omega^{(i)} = \Omega^{(i)}_{\rm freq} \times \Omega^{(i)}_{\rm time}$, with $\Omega^{(i)}_{\rm freq} = \{n \in [N] : \bs{x}_{n,t}^{(i)} \neq \bs{0}~\forall t\}$ and $\Omega^{(i)}_{\rm time} = \{t \in [T] : \bs{x}_{n,t}^{(i)} \neq \bs{0}~\forall n\}$.
Then, $\Omega = \Omega^{(0)} \setminus \bigcup_{i=1}^{O} \Omega^{(i)}$ are the time-frequency resources, used by the rUE, which have no interference, and we use the notation $M^{(i)} = |\Omega^{(i)}|$, $N^{(i)} = |\Omega^{(i)}_{\rm freq}|$, and $T^{(i)} = |\Omega^{(i)}_{\rm time}|$.
For $(n, t) \in \Omega^{(i)}$, we model the transmitted signal as independent circularly-symmetric complex Gaussian, specifically, $\bs{x}_{n,t}^{(i)} \sim {\rm CN}\big(\bs{0}, \sigma_i^2\bs{I}\big)$ where $\sigma_i^2 = E_i/(M^{(0)} N_{\rm u})$ for power $E_i > 0$.

The received signal at the r\gls{ue} in the $n$-th subcarrier, $t$-th \gls{ofdm} symbol, and $k$-th antenna is modelled as
\begin{align}
    y_{n,t,k} &= \underbrace{\sum_{l=1}^{O+L} \alpha_{0,l} g^{(0)}_{n,t,k}(\theta_{l}, \theta_{l}, 2\tau_{l})}_{{\rm rUE-iUE/SP-rUE}} + \underbrace{\sum_{i=1}^O \alpha_{i} g^{(i)}_{n,t,k}(\phi_{i}, \theta_{i}, \tau_{i})}_{{\rm iUE-rUE}} \nonumber\\
    & \quad + \underbrace{\sum_{i=1}^O \sum_{\substack{l=1\\l\neq i}}^{O+L} \alpha_{i,l} g^{(i)}_{n,t,k}(\phi_{i,l}, \theta_{l}, \tau_{i, l})}_{{\rm iUE-iUE/SP-rUE}} + \varepsilon_{n,t,k}, \label{eq:signal}
\end{align}
where $g^{(i)}_{n,t,k}(\phi, \theta, \tau) = \langle \bs{a}^*(\phi), \bs{x}_{n,t}^{(i)}\rangle d_n(\tau) a_k(\theta)$, and
$\varepsilon_{n,t,k}$ are circularly-symmetric complex \gls{awgn} with zero mean and variance $\sigma^2$.
In this model, $\alpha_i,\alpha_{0,i},\alpha_{i,l} \in \C$ are the complex path coefficients modelling path loss due to distance and scattering as well as stochastic phase shifts due to small movements or rotations. $\theta_i$ are the \gls{aoa} at the rUE from the \glspl{sp} and iUEs, $\tau_{i}$ are the \glspl{toa} of the paths from the rUE to the \glspl{sp} and iUEs, $\tau_{i,l}$ is the \gls{toa} from the $i$-th iUE to the rUE through the $l$-th scatter path, and $d_n(\tau) = {\rm exp}\big(-j2\pi \Delta_f n\tau\big)$ is the delay response. $\phi_{i,l}$ is the \gls{aod} at the $i$-th iUE to the $l$-th iUE/\gls{sp}, $\phi_{i}$ is the \gls{aod} at the $i$-th iUE to the rUE. The array response is defined as $a_k(\theta) = N_{\rm u}^{-1/2} {\rm exp}(j\pi\sin(\theta)(k-1))$, $\bs{a}(\theta) = [a_1(\theta),\dots,a_{N_u}(\theta)]^\top$ denotes the array response vector, $\bs{A}(\theta,\phi)=\bs{a}(\theta)\bs{a}^{H}(\phi)$, and $\bs{A}(\theta)=\bs{A}(\theta, \theta)$. The specifications of the channel parameters are detailed in Appendix~\ref{subsec:channel_parameters}.


\section{Interference detection and exploitation}\label{sec:methods}
The purpose of interference detection is to test the null hypotheses $H_{n,t} : (n, t) \in \Omega$, hence, a rejection of the null means that in the $n$-th subcarrier and $t$-th \gls{ofdm} symbol there is evidence of interference. This poses a multiple testing problem, and we will concern ourselves with the traditional \gls{fwer} control. Let $V$ denote the number of false rejections. The idea of \gls{fwer} control is to guarantee $\mathbb{P}(V \geq 1) \leq \delta$ for a given significance level $\delta \in (0, 1)$. In Section~\ref{subsec:interference_detection} we outline the interference detection procedure, and subsequently in Section~\ref{subsec:theoretial_validation} we provide theoretical justification of the procedure, herein deriving a simple and practical approach to \gls{fwer} control.

Upon discovering the interfered resources, we can estimate jointly the delay and angles of the targets using standard techniques with only the non-interfered resources. However, estimation of the angle to the interferer can be substantially improved by using also the interfered resources. For this reason, we determine the angles in a separate estimation step using all the resources, and finally combine the results through data association. The details are given in Section~\ref{subsec:estimation}.

\subsection{Interference detection procedure}\label{subsec:interference_detection}
The technique, proposed in \cite{Nuss2017:Novel}, works by computing the powers $\gamma_n = \sum_{t\in \Omega^{(0)}_{\rm time}} \sum_{k=1}^{N_{\rm u}} |y_{n,t,k}|^2$ for each of the $N^{(0)}$ used subcarriers. Now, the subcarriers with interference are detected as $\hat{\Omega}^{\rm int}_{\rm freq} = \{n \in \Omega^{(0)}_{\rm freq} : \gamma_n > \gamma_{(\kappa)}\beta\}$ where $\gamma_{(1)}\leq \cdots \leq \gamma_{(N^{(0)})}$ are the ordered observed powers, and $\kappa \in [N^{(0)}]$ and $\beta \geq 1$ are user-specified parameters. The idea is to choose $\kappa$ sufficiently small such that $\gamma_{(\kappa)}$ does not contain interference; a practical option is $\kappa=1$.

Following the same procedure, we can detect interfered time slots by considering the powers $\Tilde{\gamma}_t = \sum_{n\in \Omega^{(0)}_{\rm freq}} \sum_{k=1}^{N_{\rm u}} |y_{n,t,k}|^2$ and so $\hat{\Omega}^{\rm int}_{\rm time} = \{t \in \Omega^{(0)}_{\rm time} : \Tilde{\gamma}_t > \Tilde{\gamma}_{(\kappa)}\Tilde{\beta}\}$ for $\Tilde{\beta} \geq 1$. Finally, the set of non-interfered resources are estimated by removing the detected resources, $\hat{\Omega} = \Omega^{(0)} \setminus (\hat{\Omega}^{\rm int}_{\rm freq} \times \hat{\Omega}^{\rm int}_{\rm time})$.

\subsection{Theoretical justification}\label{subsec:theoretial_validation}
Under the global null, that is, when there is no interference, for each of the $N^{(0)}$ used subcarriers $y_{n,t,k} \sim {\rm CN}(0, \varsigma_{n,k}^2)$ where 
\begin{equation*}
    \begin{split}
        \varsigma_{n,k}^2 &= \sigma^2 + \sigma_0^2\Big(\sum_{l=1}^{O+L} |\alpha_{0,l}|^2 + 2|\alpha_{0,l}|\sum_{j=1}^{l-1} |\alpha_{0,j}|\\
        &\qquad \times {\rm Re}\big\{ d_n(\tau_l-\tau_j) a_k(\theta_l) a_k^*(\theta_j) \langle \bs{a}(\theta_j),\bs{a}(\theta_l)\rangle\big\}\Big).
    \end{split}
\end{equation*}
In case the objects are all well-separated by the antenna array, meaning $\langle \bs{a}(\theta_j),\bs{a}(\theta_l)\rangle \approx 0$, then the variance of $y_{n,t,k}$ is approximately constant and equal to $\varsigma^2 = \sigma^2 + \sigma_0^2\sum_{l=1}^{O+L} |\alpha_{0,l}|^2$. Hence, approximately, $|y_{n,t,k}|^2 \sim {\rm Exponential}(\varsigma^2)$. Across different time slots, $|y_{n,t,k}|^2$ are independent, however, across the antenna array they are correlated since the transmitted signal $\bs{x}_{n,t}$ is shared. Yet, if $T^{(0)} \gg N_{u}$, it holds approximately that $\gamma_1,\dots,\gamma_{N^{(0)}} \stackrel{{\rm iid}}{\sim} {\rm Gamma}(\varrho, \varsigma^2)$, for $\varrho = T^{(0)} N_u$.

Our interest in the following is to find the minimum $\beta$, given $\kappa$, such that under the global null $\mathbb{P}(\gamma_{(N^{(0)})} > \gamma_{(\kappa)}\beta) \leq \delta$, where $\delta \in (0, 1)$ is the user-specified significance level. Such a procedure controls the \gls{fwer} at level $\delta$. For $\kappa = 1$, the following result shows that $\mathbb{P}(\gamma_{(N^{(0)})} > \gamma_{(\kappa)}\beta)$ does not depend on the unknown channel, and in fact only depends on $\beta$ and the shape parameter $\varrho = T^{(0)} N_u$, that is, the number of terms in the sum of powers. Hence, finding the minimum $\beta$, denoted $\beta^*$, satisfying $\mathcal{P}(\gamma_{(N^{(0)})} > \gamma_{(1)}\beta) \leq \delta$ can be done offline, thereby presenting no practical challenge.
\begin{proposition}\label{prop:threshold}
    Let $\gamma_1,\dots,\gamma_n \stackrel{{\rm iid}}{\sim} {\rm Gamma}(\varrho, \varsigma^2)$ with order statistics $\gamma_{(1)}\leq \cdots \leq \gamma_{(n)}$. Then, $\mathbb{P}(\gamma_{(n)} > \gamma_{(1)}\beta)$ does not depend on $\varsigma^2$.
\end{proposition}
\begin{proof}
    For a proof see Appendix~\ref{sec:proof_prop1}.
\end{proof}
In Figure~\ref{fig:beta_search}, we consider a scenario with settings as described in Section~\ref{sec:numerical_experiments}. Then, comparing the theoretical and empirical \gls{fwer} for $\kappa=1$ we observe that the approximating is accurate. Further, it shows that $\beta^*=1.561$ is the minimum $\beta$ satisfying $\mathcal{P}(\gamma_{(N^{(0)})} > \gamma_{(1)}\beta) \leq \delta$ when $\delta=0.01$. Exploring \gls{fwer} control, with this procedure, for other channel models is a topic for future investigation.

Letting $\gamma_i/\gamma_{(1)}$ be a test statistic, we can also derive p-values under the approximation, $p_i=1-F_{\gamma_i/\gamma_{(1)}}(z)$ where $F_{\gamma_i/\gamma_{(1)}}(\cdot)$ is the \gls{cdf} for the random variable $\gamma_i/\gamma_{(1)}$ and $z$ is the observed test statistic. Again the scale parameter vanishes:
\begin{equation}\label{eq:p-values}
    p_i = (n-1) \int_{0}^{1} \big(1-F_\varrho\big(zF_\varrho^{-1}(1-w)\big)\big) w^{n-2} {\rm d}w,
\end{equation}
where $f_\varrho(\cdot)$ and $F_\varrho(\cdot)$ are the \gls{pdf} and \gls{cdf} of a Gamma distributed random variable with shape parameter $\varrho$ and unit rate, respectively. Unfortunately, these p-values are jointly dependent through $\gamma_{(1)}$ complicating the validity of multiplicity correction procedures, and further, evaluating the integral in Eq.~\eqref{eq:p-values} can be prohibitive in practical deployments.

\begin{figure}[t]
    \centering
    \includegraphics[width=0.75\linewidth, page=2]{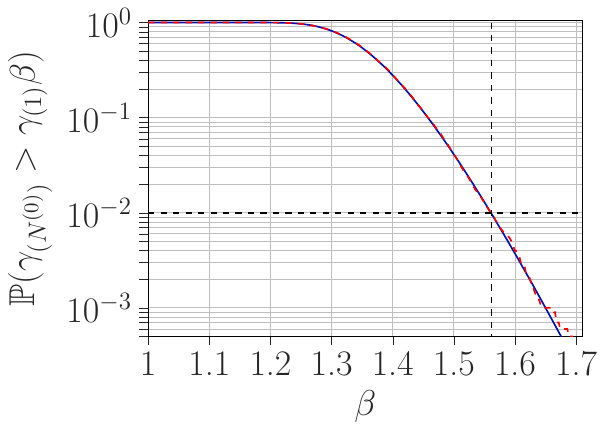}
    \caption{\gls{fwer} under the global null for varying $\beta$. The solid blue line and the dashed red line is the theoretical and empirical \gls{fwer}, respectively.}
    \label{fig:beta_search}
\end{figure}

\begin{figure*}
\normalsize
\setcounter{mytempeqncnt}{\value{equation}}
\begin{equation}\label{eq:cov}
    \bs{\Sigma}_{n,t} = \sigma^2 \bs{I} + \sum_{i=1}^{O} \sigma^2_i \mathbbm{1}[(n,t)\in\Omega^{(i)}] \Big(\bs{A}(\theta_i) |\alpha_i|^2 + \sum_{\substack{l=1\\l\neq i}}^{O+L} \bs{A}(\theta_{l}) |\alpha_{i,l}|^2 + \bs{V}_n^{i,l} + (\bs{V}_n^{i,l})^{H} + \sum_{j=1}^{l-1} \bs{W}_{n}^{i,l,j} + (\bs{W}_{n}^{i,l,j})^{H}\Big),
\end{equation}
\begin{equation*}
    \text{where } \bs{V}_n^{i,l} = \bs{A}(\theta_i,\theta_{l}) \alpha_i\alpha_{i,l}^* \langle \bs{a}(\phi_{i,l}), \bs{a}(\phi_i)\rangle d_n(\tau_i - \tau_{i,l}),
    \text{ and } \bs{W}_{n}^{i,l,j} = \bs{A}(\theta_{l},\theta_{j}) \alpha_{i,l} \alpha_{i,j}^* \langle \bs{a}(\phi_{i,j}), \bs{a}(\phi_{i,l}) \rangle d_n(\tau_{i,l} - \tau_{i,j}).
\end{equation*}
\setcounter{equation}{\value{mytempeqncnt}}
\hrulefill
\vspace*{-10pt}
\end{figure*}

\subsection{Delay and angle estimation}\label{subsec:estimation}
The proposed approach to delay and angle estimation consists of three steps: (i) estimate angles $\hat{\theta}_j^{\rm init}$ using the \gls{music} algorithm with all the $\Omega^{(0)}$ resources \cite{Stoica1989:MUSIC}; (ii) estimate jointly delays and angles $(\hat{\tau}_i, \hat{\theta}_i)$ using the \gls{omp} algorithm with the non-interfered resources $\hat{\Omega}$ \cite{Tropp2007:OMP}; (iii) associate estimated parameters $\hat{\theta}_j^{\rm init}$ with $(\hat{\tau}_i, \hat{\theta}_i)$. We use the \gls{music} algorithm in the first step since \gls{omp} fails in the presence of an unknown interference signal, meanwhile \gls{music} is robust against this for angle estimation. \rev{We assume the number of targets are perfectly estimated for instance using \cite{Anderson1963:Asymptotic}.}


\emph{(i) MUSIC:}
The covariance matrix across the antenna array is estimated as $\hat{\bs{\Gamma}} = \sum_{(n,t)\in\Omega^{(0)}} \bs{y}_{n,t}\bs{y}_{n,t}^{H}$ where $\bs{y}_{n,t} = [y_{n,t,1}, \dots, y_{n,t,N_u}]^\top$. The eigenvectors of $\hat{\bs{\Gamma}}$ are computed and ordered according to the eigenvalues, denoted $\bs{u}_1,\dots,\bs{u}_{N_u}$. A basis for the noise subspace is then $\bs{U}_n = [\bs{u}_1,\dots,\bs{u}_{N_u-s}]$ where $s$ is the number of impinging signals. Finally, $\hat{\theta}_i^{\rm init}$ are determined as the $s$ largest local maxima in the \gls{music} pseudo spectrum \cite{Stoica1989:MUSIC}:
\begin{equation*}
    P_{\rm music}(\theta) = \frac{1}{\bs{a}^{H}(\theta)\bs{U}_n\bs{U}_n^{H}\bs{a}(\theta)}.
\end{equation*}

\emph{(ii) Orthogonal matching pursuit:}
Denote by $\bs{h}(\theta, \tau) = [\langle \bs{a}^*(\phi), \bs{x}_{n,t}^{(0)}\rangle d_n(\tau) \bs{a}^\top(\theta)]_{(n,t) \in \hat{\Omega}}$, and let the so-called residual be initialized as $\bs{r}_0 = \bs{y}^{\rm est} = [\bs{y}^\top_{n,t}]_{(n,t) \in \hat{\Omega}}$. \gls{omp} commences by solving
\begin{equation*}
    \hat{\tau}_1, \hat{\theta}_1 = \argmax_{\tau, \theta}  \big|\langle \bs{h}^*(\theta, \tau), \bs{r}_0\rangle\big|,
\end{equation*}
and then defining $\bs{\Psi}_1 = [\bs{h}(\hat{\tau}_1, \hat{\theta}_1)]$. Then, the residual is updated as
\begin{equation*}
    \bs{r}_1 = \bs{y}^{\rm est} - \bs{\Psi}_1 (\bs{\Psi}_1^{H} \bs{\Psi}_1)^{-1} \bs{\Psi}_1^{H} \bs{y}^{\rm est},
\end{equation*}
and the algorithm proceeds to find the next pair of channel parameters $\hat{\tau}_2$ and $\hat{\theta}_2$, and so on. The algorithm proceeds until $s$ pairs of channel parameters have been determined \cite{Tropp2007:OMP}.

\emph{(iii) Data association:}
We assign a cost
\begin{equation*}
    \bs{C}_{i, j} = {\rm min}\big((\hat{\theta}_i - \hat{\theta}_j^{\rm init})^2, (\lvert\hat{\theta}_i - \hat{\theta}_j^{\rm init}\rvert - 2\pi)^2\big)
\end{equation*}
for the association of $\hat{\theta}_j^{\rm init}$ and $(\hat{\tau}_i, \hat{\theta}_i)$.
Let $\bs{X}\in\{0,1\}^{s\times s}$ be the boolean assignment matrix such that $\bs{X}_{i,j} = 1$ if $(\hat{\tau}_i, \hat{\theta}_i)$ is associated with $\hat{\theta}_j^{\rm init}$, and $\bs{X}_{i,j} = 0$ otherwise. The data association problem can then be expressed by the following optimization problem
\begin{align*}
    {\rm minimize}~~&~ \text{Tr}(\bs{X}^\top \bs{C})\\
    {\rm subject~to}~~&~\bs{X}_{i,j}\in\{0,1\},~\forall i,j,\\
    &\sum_{i=1}^{s} \bs{X}_{i,j} = 1,~\forall j, \:~~ \sum_{j=1}^s \bs{X}_{i,j} = 1,~\forall i,
\end{align*}
solvable with standard algorithms \cite{Murty1968:Algorithm}.
The final parameter pairs are $(\hat{\tau}_i, \hat{\theta}_i^{\rm upd})$ where $\hat{\theta}_i^{\rm upd} = \hat{\theta}_i + \rev{\bs{X}_{i, w}}(\hat{\theta}_{w}^{\rm init}-\hat{\theta}_i)$ for $i\in[s]$, and $w = \argmax_{j\in[s]} P_{\rm music}(\theta_j^{\rm init})$ is an index associated with an interferer.
%

\section{Cram\'{e}r-Rao lower bound}\label{sec:theoretical_performance_bounds}
In this section, we present the \gls{crlb} on the estimation of the \gls{toa} and \gls{aoa} of the targets in the presence of interference. For conciseness, we do not provide analytical expressions for the \gls{crlb}, but expose how to numerically compute the \gls{crlb}.

Recall the signal model Eq.~\eqref{eq:signal} and notice that $\bs{y}_{n,t} \sim {\rm CN}\big(\bs{\mu}_{n,t}, \bs{\Sigma}_{n,t}\big)$ where
\begin{equation*}
    \bs{\mu}_{n,t} = \sum_{i=1}^{O+L} \alpha_{0,i} \langle \bs{a}^*(\theta_i), \bs{x}_{n,t}^{(0)}\rangle d_n(2\tau_{i}) \bs{a}(\theta_i),
\end{equation*}
and $\bs{\Sigma}_{n,t}$ is given in Eq.~\eqref{eq:cov}.
The Fisher information matrix is then given by \cite[Sec.~B.3]{Stoica2005:Spectral}
\begin{equation*}
    \begin{split}
        [\bs{F}_{n,t}]_{i,j} &= {\rm Tr}\Big(\bs{\Sigma}_{n,t}^{-1} \frac{\partial \bs{\Sigma}_{n,t}}{\partial \bs{\eta}_{i}} \bs{\Sigma}_{n,t}^{-1} \frac{\partial \bs{\Sigma}_{n,t}}{\partial \bs{\eta}_{j}}\Big) \\
        &\qquad\quad + 2{\rm Re}\Big\{\frac{\partial \bs{\mu}^H_{n,t}}{\partial \bs{\eta}_i} \bs{\Sigma}_{n,t}^{-1} \frac{\partial \bs{\mu}_{n,t}}{\partial \bs{\eta}_j}\Big\},
    \end{split}
\end{equation*}
where $\bs{\eta} = [\Tilde{\bs{\eta}}^\top, \bar{\bs{\eta}}^\top]^\top$ for parameters of interest
\begin{equation*}
    \Tilde{\bs{\eta}} = [\tau_{1}, \dots, \tau_{O+L}, \theta_1, \dots, \theta_{O+L}]^\top,
\end{equation*}
and nuisance parameters
\begin{equation*}
    \bar{\bs{\eta}} = [{\rm Re}\{\alpha_{0,1}\}, \dots, {\rm Re}\{\alpha_{O,L}\}, {\rm Im}\{\alpha_{0,1}\}, \dots, {\rm Im}\{\alpha_{O,L}\}]^\top.
\end{equation*}
For conciseness of presentation, we omit the details of the partial derivatives entering in the Fisher information matrix, remarking that the expressions are computed through straightforward calculus \cite{AbuShaban2018Bounds,Zheng2024:JrCUP,Vejling2025:Resolution}.

The Fisher information matrix accumulating information from all \gls{ofdm} symbols and subcarriers is $\bs{F}_{i,j} = \sum_{(n,t)\in \Omega^{(0)}} [\bs{F}_{n,t}]_{i,j}$, and from this we compute the Cram\'{e}r-Rao lower bound as $\bs{K} = \bs{F}^{-1}$. Here, $\bs{K}$ is interpreted as the lower bound on the covariance matrix of the estimator of $\bs{\eta}$ given the signal model. To compare performance for angle and delay estimation, we will use the \gls{aeb} and \gls{deb} defined as ${\rm AEB}_l = \sqrt{\bs{K}_{O+L+l, O+L+l}}$ and ${\rm DEB}_l = \sqrt{\bs{K}_{l, l}}$, respectively.




\section{Numerical experiments}\label{sec:numerical_experiments}
We define a simple scenario with a single iUE and a single \gls{sp} located at $[5, 14]$ m and $[17, 6]$ m, respectively. Each device has $N_u = 6$ antennas, and the \gls{ofdm} resource block consists of $T=30$ time slots each with a total of $N=64$ subcarriers. The rUE uses all time slots, but only $32$ spectrally interleaved subcarriers, selected random uniformly. Similarly, the iUE also uses all time slots. The carrier frequency is $f_c=15$ GHz with a subcarrier spacing of $\Delta_f = 250$ kHz. The noise power is $\sigma^2=-173.85$ dBm. In the numerical experiments, we will vary the number of interfered subcarriers, and the transmission powers. The scenario settings are summarized in Table~\ref{tab:params}.

\setlength{\tabcolsep}{0.1em}
\renewcommand{\arraystretch}{1.1}
\begin{table}[t]
    \centering
    \caption{Scenario settings}\label{tab:params}
    \begin{tabular}{ccc}
        \toprule
        \textbf{Parameter}	        & \textbf{Value}    & \textbf{Description} \\ \midrule
        $N_u$                       & 6                 & Number of antennas. \\
        $T$                         & 30                & Number of OFDM symbols. \\
        $N$                         & 64                & Number of subcarriers. \\
        $T^{(0)}$ & 30              & Number of OFDM symbols used by the rUE. \\
        $N^{(0)}$ & 32              & Number of subcarriers used by the rUE. \\
        $f_c$                       & 15 GHz            & Carrier frequency. \\
        $\Delta_f$                  & 250 kHz           & Subcarrier spacing. \\
        $\sigma^2$                  & $-173.85$ dBm     & Noise power. \\
        $\bs{p}_0$                  & $[0, 0]$ m        & Position of rUE. \\
        $\bs{p}_1$                  & $[5, 14]$ m       & Position of iUE. \\
        $\bs{p}_2$                  & $[17, 6]$ m       & Position of SP. \\
        \bottomrule
    \end{tabular}
\end{table}
\begin{figure}[t]
    \centering
    \includegraphics[width=0.67\linewidth, page=5]{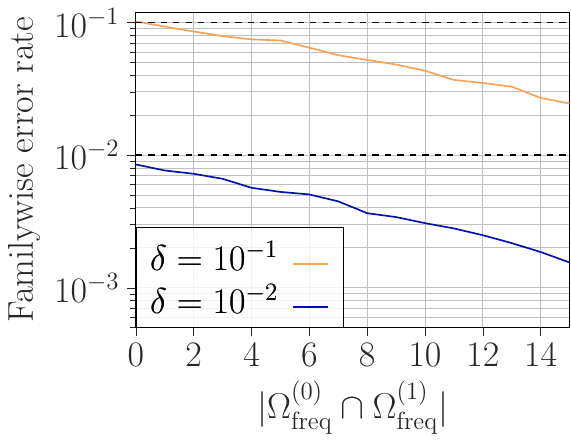}
    \caption{\gls{fwer} when $\kappa=1$ and $\delta\in\{10^{-1},10^{-2}\}$ for varying number of overlapping subcarriers.}
    \label{fig:familywise_error_rate}
\end{figure}
\begin{figure}[t]
    \centering
    \includegraphics[width=0.67\linewidth, page=3]{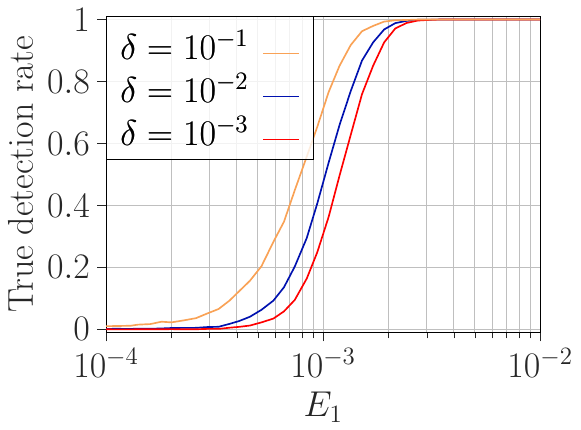}
    \caption{True detection rate when $\kappa=1$, $\delta\in\{10^{-1},10^{-2},10^{-3}\}$, and $E_0=0.05$ for varying interference powers, $E_1$.}
    \label{fig:true_detection_rate}
\end{figure}

Two baseline techniques are considered as benchmarks:
\emph{(1) Naïve}: This baseline ignores the issues of interference, performing sensing using \gls{omp} with all the used resources.\\
\emph{(2) Oracle}: This baseline knows which resources have interference and otherwise proceeds as the proposed algorithm.

We compare the proposed method to the baselines using the \gls{rmse}. Denote by $\hat{\tau}_l$ and $\hat{\theta}_l$ the delay and angle estimators for the $l$-th target. Then, the \gls{rmse} for delay and angle estimation is ${\rm RMSE}_l^{\rm delay} = \mathbb{E}[(\hat{\tau}_l - \tau_l)^2]^{1/2}$ and ${\rm RMSE}_l^{\rm angle} = \mathbb{E}[(\hat{\theta}_l - \theta_l)^2]^{1/2}$, respectively.


First, let us validate the \gls{fwer} guarantees, by reporting the empirical \gls{fwer} of the interference detection algorithm in Figure~\ref{fig:familywise_error_rate}, when $\kappa=1$, $\delta\in\{10^{-1},10^{-2}\}$, and for varying number of overlapping subcarriers. We observe that in all cases, the \gls{fwer} is upper bounded as specified by the theory. Moreover, we notice that as the number of overlapping subcarriers increase, the procedure exhibits some excess conservativeness, and it is a direction of future study to avoid this. Next, let us consider the statistical power of the interference detection. Figure~\ref{fig:true_detection_rate} displays the true detection rate when $\kappa=1$, $\delta\in\{10^{-1},10^{-2},10^{-3}\}$, and $E_0=0.05$ for varying interference powers, $E_1$. This result shows that due to the relatively high power of the direct iUE-rUE link, as compared to the backscattered signals, interference is consistently detected, even when the interference power is relatively low and with a strict \gls{fwer} restriction of $\delta=10^{-3}$.

\begin{figure}[t]
    \centering
    \includegraphics[width=0.99\linewidth, page=4]{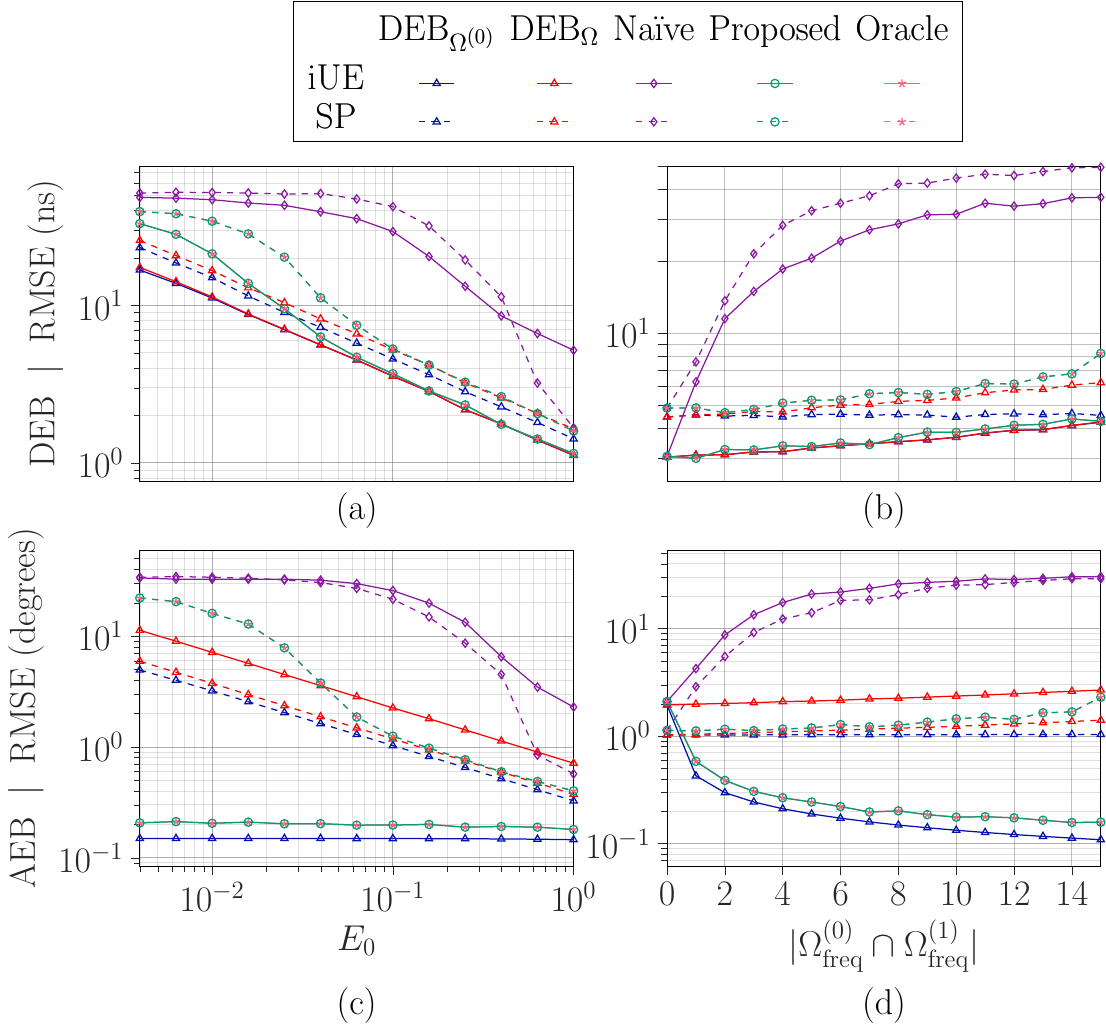}
    \caption{(a),(c) Comparison for varying transmission powers, $E_0$, when the number of overlapping subcarriers are 8. (b),(d) Comparison for varying number of overlapping subcarriers, when the transmission power is $E_0=0.1$. Settings: $\kappa=1$, $\delta=10^{-3}$, and $E_1=0.05$.}
    \label{fig:main}
\end{figure}

The main result is given in Figure~\ref{fig:main}: \glspl{deb} and \glspl{aeb} are compared to the \gls{rmse} with the proposed algorithm and the benchmarks for each of the targets, with $\kappa=1$, $\delta=10^{-3}$, and $E_1=0.05$.
We include different theoretical bounds: ${\rm DEB}_{\Omega^{(0)}}$ and ${\rm AEB}_{\Omega^{(0)}}$ which are the error bounds with all the used resources, and ${\rm DEB}_{\Omega}$ and ${\rm AEB}_{\Omega}$ giving the bounds when only using the resources without interference.
Theoretically, ${\rm DEB}_{\Omega^{(0)}}$ and ${\rm AEB}_{\Omega^{(0)}}$ are guaranteed to be lower than ${\rm DEB}_{\Omega}$ and ${\rm AEB}_{\Omega}$, respectively, however as Figure~\ref{fig:main} show, the improvement by using the resources with interference is only non-negligble for the angle to the iUE. Further, Figure~\ref{fig:main} illustrates the poor performance when not correcting for the interference with the naïve baseline, and also indicates that the proposed method is capable of (nearly) reaching the theoretical performance bounds when the transmit power $E_0$ is sufficiently large. In Figures~\ref{fig:main}(a),(c), the bounds are approached at $E_0=0.06$. Finally, note that the performance of proposed algorithm coincides with that of the oracle, demonstrating the statistical power of the interference detection technique and the robustness to varying conditions.


\section{Conclusion}\label{sec:conclusion}
This paper addressed the problem of mutual interference in spectrally interleaved OFDM-based ISAC systems with radar sensing. We derived an interference detection method that provably controls the familywise error rate, and proposed a channel parameter estimation technique exploiting interference for angle estimation while avoiding it for delay estimation. A tailored Fisher analysis confirmed that interfered resources are only useful for estimating the angle to the interferer, and numerical results demonstrated that the proposed algorithm accurately detects interference and achieves near-optimal delay–angle estimation. Future work includes extending the framework to also consider communication aspects and exploring multi-user tracking scenarios with adaptive resource access policies.


\appendices

\section{Channel parameters}\label{subsec:channel_parameters}
The complex path gains, denoted by $\alpha_{i}, \alpha_{i,l}$, correspond to the \gls{los}-path repetitively. scattering paths, respectively. They are modeled as $\alpha_{i}=|\alpha_{i}|{\rm exp}(j\nu_i)$, and $\alpha_{i,l}=|\alpha_{i,l}|{\rm exp}(j\nu_{i,l})$, where $\nu_{i},\nu_{i,l}\sim{\rm Uniform}[0,2\pi)$ and $|\alpha_{i}|^2=\lambda^2(4\pi)^{-2}\|\mathbf{p}_{i} \|^{-2}$, $|\alpha_{i,l}|^2=\lambda^2(4\pi)^{-3} \| \mathbf{p}_{l}- \mathbf{p}_{i} \|^{-2} \| \mathbf{p}_{l} \|^{-2}$.
The \gls{aoa}, \gls{aod} and \gls{toa} are modelled as \cite{AbuShaban2018Bounds}: $\theta_i=\arctantwo(p_{i,y},p_{i,x})$, $\phi_{i}=\mathbbm{1}[\theta_i \geq 0](\theta_i-\pi) + \mathbbm{1}[\theta_i < 0](\theta_i+\pi)$, $\tau_{i}=\|\bs{p}_{i}\|/c$, $\phi_{i,l}=\arctantwo(p_{l,y}-p_{i,y},p_{l,x}-p_{i,x})$, and $\tau_{i,l}=\tau_{l} + \|\bs{p}_{l}-\bs{p}_i\|/c$.

\section{Proofs of Proposition 1}\label{sec:proof_prop1}
Before proceeding to the proof of Proposition~\ref{prop:threshold}, we present the following technical lemma.
\begin{lemma}\label{lem:ratio_dist}
    Let $\gamma_1,\dots,\gamma_n$ be iid random variables with \gls{pdf} $f$ and \gls{cdf} $F$. Denote by $\gamma_{(1)}\leq \cdots \leq \gamma_{(n)}$ the order statistics. Then, for $\beta \geq 1$.
    \begin{align*}
        &\mathbb{P}(\gamma_{(n)} > \gamma_{(1)}\beta) \\
        &\:= n(n-1)\int_0^\infty \int_{\beta u}^\infty f(u)f(v)\big(F(v)-F(u)\big)^{n-2} {\rm d}v{\rm d}u.
    \end{align*}
\end{lemma}
\begin{proof}[Proof of Lemma 1]\label{proof:proof_lemma1}
    First we derive the \gls{cdf} for $(\gamma_{(n)}, \gamma_{(1)})$:
    \begin{align*}
        F_{\gamma_{(n)}, \gamma_{(1)}}(v,u) &= \mathbb{P}(\gamma_{(n)} \leq v, \gamma_{(1)} \leq u)\\
        \begin{split}
        &=\mathbb{P}(\gamma_1 \leq v, \dots, \gamma_n \leq v)\\
        &\qquad - \mathbb{P}(u < \gamma_1 \leq v, \dots, u < \gamma_n \leq v)\\
        \end{split}\\
        &= F(v)^n - (F(v) - F(u))^n,
    \end{align*}
    for $v \geq u$.
    Now, by differentiation we find the \gls{pdf}:
    \begin{align*}
        f_{\gamma_{(n)}, \gamma_{(1)}}(v,u) &= \frac{{\rm d}}{{\rm d}v} \frac{{\rm d}}{{\rm d}u} F_{\gamma_{(n)}, \gamma_{(1)}}(v,u)\\
        &= n (n-1) f(u) f(v) (F(v) - F(u))^{n-2},
    \end{align*}
    for $v \geq u$. Finally, for $\beta \geq 1$
    \begin{equation*}
        \mathbb{P}(\gamma_{(n)} > \gamma_{(1)} \beta) = \int_0^\infty \int_{\beta u}^\infty f_{\gamma_{(n)}, \gamma_{(1)}}(v,u) {\rm d}v{\rm d}u.
    \end{equation*}
\end{proof}
\vspace{-12pt}
\begin{proof}[Proof of Proposition 1]\label{proof:proof_prop1}
    By Lemma~\ref{lem:ratio_dist}, for $\beta \geq 1$,
    \begin{equation*}
        \begin{split}
            &\mathbb{P}(\gamma_{(n)} > \gamma_{(1)}\beta) = \frac{n(n-1)}{\varsigma^{4\varrho} (\varrho-1)!^n} \int_0^\infty \int_{\beta u}^\infty \Big((uv)^{\varrho-1} \\
            &\quad \times {\rm exp}\Big(-\frac{u+v}{\varsigma^2}\Big) \Big(\int_{u/\varsigma^2}^{v/\varsigma^2} t^{\varrho-1} {\rm exp}(-t) {\rm d}t\Big)^{n-2} {\rm d}v {\rm d}u\Big).
        \end{split}
    \end{equation*}
    By a variable substitution $x=u/\varsigma^2$ and $y=v/\varsigma^2$
    \begin{equation*}
        \begin{split}
            &\mathbb{P}(\gamma_{(n)} > \gamma_{(1)}\beta) = \frac{n(n-1)}{(\varrho-1)!^n} \int_0^\infty \int_{\beta x}^\infty \Big((xy)^{\varrho-1} \\
            &\quad \times {\rm exp}(-(x+y)) \Big(\int_{x}^{y} t^{\varrho-1} {\rm exp}(-t) {\rm d}t\Big)^{n-2} {\rm d}y {\rm d}x\Big).
        \end{split}
    \end{equation*}
\end{proof}
\vspace{-10pt}


\ifCLASSOPTIONcaptionsoff
  \newpage
\fi

\bibliographystyle{IEEEtran}
\bibliography{bibliography}

\end{document}